\documentclass[journal]{IEEEtran}
\usepackage{amsmath,amssymb,amsthm}    
\usepackage[dvips]{graphicx}   
\usepackage{verbatim}   
\usepackage{color}      
\usepackage{subfigure}  
\usepackage{epsfig}
\usepackage{float}
\usepackage{algorithm}
\usepackage{algorithmic}
\usepackage{setspace}
\usepackage{subfigure}
\usepackage{cite}
\newtheorem{lemma}{\bf Lemma}
\newtheorem{theorem}{\bf Theorem}

\newtheorem{proposition}{\bf Proposition}

\makeatletter
\newcommand*{\rom}[1]{\expandafter\@slowromancap\romannumeral #1@}
\makeatother
\begin{document}
\title{\textbf{Measurement Matrix Design for Compressive Detection with Secrecy Guarantees}}
\author{Bhavya~Kailkhura,~\IEEEmembership{Student Member,~IEEE}, Sijia~Liu,~\IEEEmembership{Student Member,~IEEE},
Thakshila~Wimalajeewa,~\IEEEmembership{Member,~IEEE}, Pramod~K.~Varshney,~\IEEEmembership{Fellow,~IEEE}
\thanks{This work was supported in part by the National Science Foundation (NSF) under Grant No. 1307775.}
\thanks{Authors are with Department of EECS, Syracuse University, Syracuse, NY 13244. (email: bkailkhu@syr.edu; sliu17@syr.edu; twwewelw@syr.edu; varshney@syr.edu)}}
\date{}
\maketitle
\begin{abstract} 
In this letter, we consider the problem of detecting a high
dimensional signal based on compressed
measurements with physical layer secrecy guarantees. 
We assume that the network operates
in the presence of an eavesdropper who intends to discover
the state of the nature being monitored by the system.
We design measurement matrices which maximize the detection performance of the network while guaranteeing a certain level of secrecy.
We solve the measurement matrix design problem under three different scenarios: $a)$ signal is known, $b)$ signal lies in a low dimensional subspace, and $c)$ signal is sparse. 
It is shown that the security performance of the system can be improved by using optimized measurement matrices
along with artificial noise injection based techniques.
 
\end{abstract}

\section{Introduction}
The theory of compressive sensing (CS)  mostly deals with reconstructing  a sparse signal based on a small number of measurements obtained via low dimensional projections~\cite{donoho,candes, cssurvey}. The use of CS based measurement schemes in solving inference problems, such as, detection, classification and estimation has also attracted a considerable attention in the recent literature~\cite{dave,haupt,thak,durate}. 
The random measurement scheme used in~\cite{dave,haupt,thak,durate} provides universality for a wide variety of signal classes, but it fails to exploit the signal structure that may be known \textit{a priori}. To improve performance, optimization of the measurement scheme can be performed by exploiting the signal structure. Measurement matrix design in the context of sparse signal recovery from compressed measurements has been considered in the past~\cite{elad, abol, gang, endra, rodrigues}. In~\cite{Zahedi}, the authors considered the measurement matrix deign problem in the context of sparse signal detection based on compressed measurements. In all of these works~\cite{Zahedi, elad, abol, gang, endra, rodrigues}, measurement matrices were constrained to the class of tight frames in order to avoid coloring the noise covariance matrix. In~\cite{subspace}, the measurement matrix was designed for the detection of a known sparse signal from compressive measurements. In the same context, heuristic or algorithmic approaches were proposed by~\cite{bai} for measurement matrix design.

In this letter, our goal is to design measurement matrices for compressive detection so that a desired performance level is achieved under physical layer secrecy constraints. To that end, we employ the collaborative compressive detection (CCD) framework proposed in our previous paper~\cite{bhavyaphy}. In a CCD framework, the fusion center (FC) receives compressed observation vectors from the nodes and makes the global decision about the presence of the signal vector. The transmissions by the nodes, however, may be observed by an eavesdropper. The secrecy of a detection system against eavesdropping attacks is of utmost importance in many applications~\cite{physecdet}.
Recently, a few attempts have been made to address the problem of eavesdropping threats on distributed detection networks~\cite{bkeve}.  
Security issues with CS based measurement schemes have been considered in~\cite{cssec1,cssec2,cssec3}, where performance limits of secrecy of CS based measurement schemes were analyzed (under different assumptions). In this work, we investigate the problem from a design perspective and consider the problem of measurement matrix design with secrecy guarantees in an optimization framework. 
We show that the performance of the CCD framework can be significantly improved by using optimized measurement matrices (which exploit the underlying signal structure) along with artificial noise injection based techniques.  
More specifically, we design optimal measurement matrices which maximize the detection performance of the network while guaranteeing a certain level of secrecy considering three different scenarios: $1)$ signal of interest $s$ is known, $2)$ $s$ lies in low dimensional subspace, and $3)$ $s$ is sparse.

\section{Observation Model}
\label{sec2}
\begin{figure}[t!]
  \centering
    \includegraphics[height=0.2\textheight, width=0.35\textwidth]{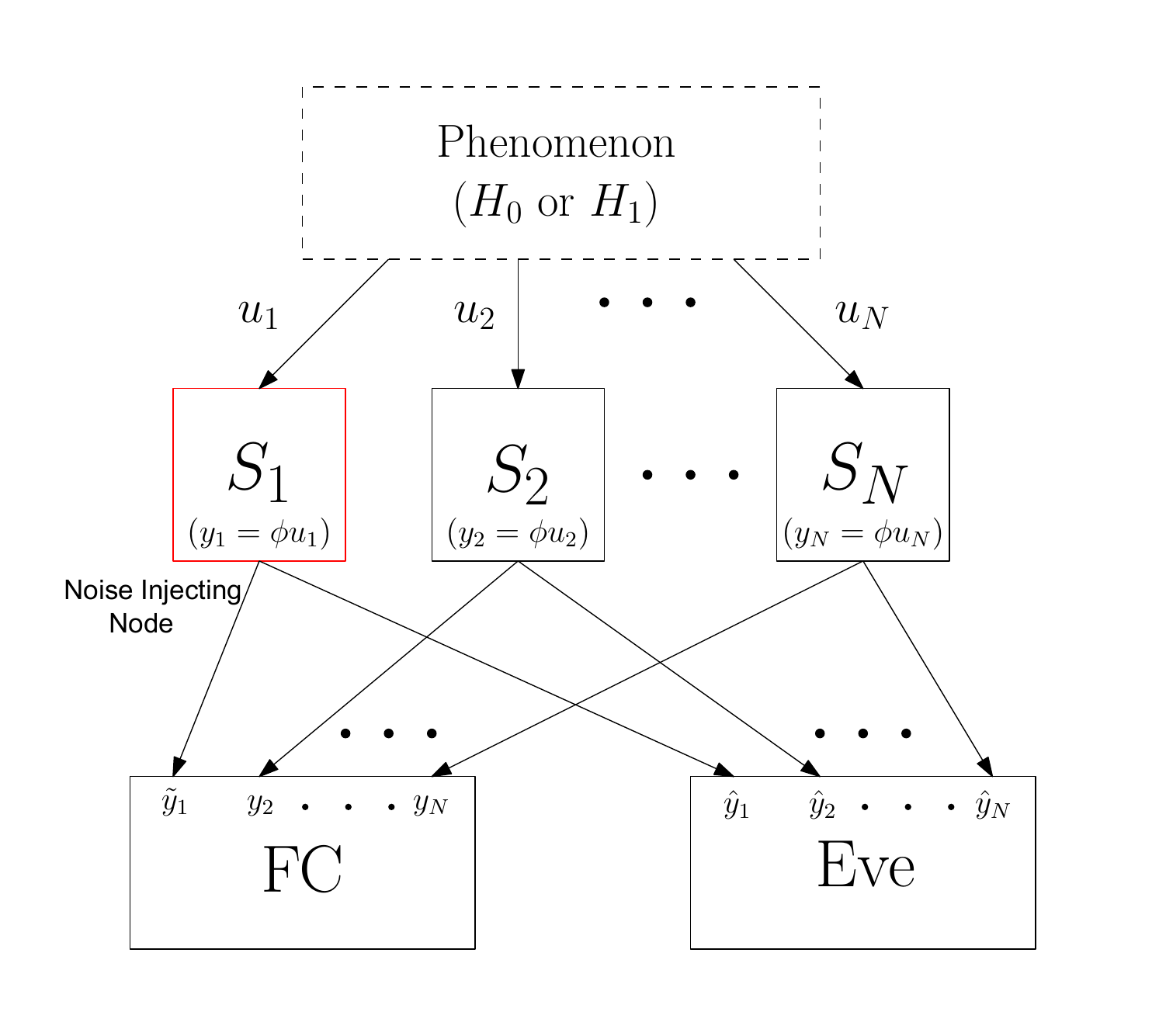}
    \vspace{-0.2in}
    \caption{System Model}
    \label{model}
    \vspace{-0.2in}
\end{figure}
\subsection{Collaborative Compressive Detection}
Consider two hypotheses  $H_{0}$ (signal is absent) and $H_{1}$ (signal is present).
Also, consider a parallel network (see Figure~\ref{model}), comprised of a fusion center (FC)
and a set of $N$ nodes, which
faces the task of determining which of the two hypotheses is true.
At the $i$th node, the observed signal, $u_i$, can be modeled as
\vspace{-0.3in}
\begin{center}
\begin{eqnarray*}
&& H_0\; :\quad u_i=v_i\\
&& H_1\; :\quad u_i=s+v_i
\end{eqnarray*}
\end{center}
for  $i =1,\cdots,N$, where $u_i$ is the $P \times 1$ observation vector, $s$ is the signal vector to be detected, $v_i$ is the additive Gaussian noise with $v_i \sim \mathcal{N} (0,\sigma^2 I_P)$ where $I_P$ is the $P\times P$ identity matrix.
Observations at the nodes are assumed to be conditionally independent and identically distributed.
Each node carries out local compression (low dimensional projection) and sends a $M$-length compressed version $y_i$ of its $P$-length observation $u_i$ to the FC. The
collection of $M$-length $(< P)$ sampled observations is given by, $y_i=\phi u_i$, where $\phi$ is an $M \times P$ measurement matrix, which is assumed to be the same for all the nodes, and $y_i$ is the $M\times 1$ compressed observation vector (local summary statistic). Under the two hypotheses, the compressed measurements are
\vspace{-0.3in}
\begin{center}
\begin{eqnarray*}
&& H_0\; :\quad y_i=\phi v_i\\
&& H_1\; :\quad y_i=\phi s+ \phi v_i
\end{eqnarray*}
\end{center}
for $i = 1,\cdots, N$. The FC makes the global decision about the phenomenon based on the received compressed measurement vectors, $\mathbf{y}=[y_1,\cdots,y_N]$.
We also assume that there is an eavesdropper present in the network who intends to discover the state of the nature being monitored by the system. 

\subsection{Collaborative Compressive Detection in the Presence of an Eavesdropper}
To keep the data regarding the presence of the phenomenon secret from the eavesdropper, in our previous work~\cite{bhavyaphy}, we used cooperating trustworthy nodes that assist the FC by injecting noise in the signal sent to the eavesdroppers to improve the security performance of the system. It was assumed that $B$ out of $N$ nodes (or $\alpha=B/N$ fraction of the nodes) tamper their data $y_i$ and send $\tilde{y_i}$ as follows: 

\noindent Under $H_0$: 
\vspace{-0.07in}
\[ \tilde{y_{i}} =  \left\{ \begin{array}{rll}
				 \phi(v_{i}+D_{i})  & \mbox{with probability}\ P_{1}^0 \\
				 \phi(v_{i}-D_{i})  & \mbox{with probability}\ P_{2}^0 \\
			     \phi v_{i}  & \mbox{with probability}\ (1-P_1^0-P_2^0)\\
				\end{array}\right.  
\]

\noindent Under $H_1$: 
\vspace{-0.07in}
\[ \tilde{y_{i}} =  \left\{ \begin{array}{rll}
				 \phi(s+v_{i}+D_{i})  & \mbox{with probability}\ P_{1}^1 \\
				 \phi(s+v_{i}-D_{i})  & \mbox{with probability}\ P_{2}^1 \\
			     \phi(s+v_{i})  & \mbox{with probability}\ (1-P_1^1-P_2^1)\\
				\end{array}\right.  
\]

\noindent where $D_{i}=\gamma s$ is a $P \times 1$ vector with constant values. The parameter $\gamma>0$ represents the noise strength, which is zero for non noise injecting nodes. We assume that the observation model and noise injection parameters are known to both the FC and the eavesdropper. The only information unavailable at the eavesdropper is the identity of the noise injecting nodes. Thus, the eavesdropper considers each node $i$ to be injecting noise with a certain probability $\alpha$. The FC can distinguish between $y_i$ and $\tilde{y_i}$. 
Note that, the values of $(P_1^0, P_2^0)$ and $(P_1^1, P_2^1)$ are system dependent and cannot be optimized in many scenarios which limits the secrecy performance of the system. Thus, in this letter, we design optimal measurement matrix $\phi$ which can be used along with artificial noise injection based techniques to improve the security performance of the system.

\section{Problem Formulation}
We use the deflection coefficient as the detection performance metric in lieu of the probability of error of the system. Deflection coefficient reflects the output signal to noise ratio and is widely used in optimizing the performance of detection systems. The deflection coefficient at the $i$th node is defined as
\vspace{-0.1in}
\begin{eqnarray*}
D(y_i)&=& (\mu_1^i-\mu_0^i)^T (\Sigma_0^i)^{-1} (\mu_1^i-\mu_0^i) 
\end{eqnarray*}
where $\mu_j^i$ and $\Sigma_j^i$ are the mean and the covariance matrix of $y_i$ under the hypothesis $H_j$, respectively. Using these notations, the deflection coefficient at the FC can be written as 
 $D(FC)=B D(\tilde{y_i})+(N-B)D(y_i).$
Dividing both sides of the above equation by $N$, we get
 $D_{FC}=\alpha D(\tilde{y_i})+(1-\alpha)D(y_i)$
where $D_{FC}={D(FC)}/{N}$ and will be used as the performance metric.
Similarly, the deflection coefficient at the eavesdropper can be written as $D_{EV}={D(EV)}/{N}=D(\hat{y_i}).$
Notice that both $D_{FC}$ and $D_{EV}$ are functions of the measurement matrix $\phi$ and noise injection parameters $(\alpha,\gamma)$ which are under the control of the FC. This motivates us to design the optimal measurement matrix for fixed noise injection parameters $(\alpha,\gamma)$
under a physical layer secrecy constraint. The problem can be formally stated as: 
\vspace{-0.1in}
\begin{equation}
\label{opt}
\begin{aligned}
& \underset{\phi}{\text{maximize}}
& & \alpha D(\tilde{y_i})+(1-\alpha)D(y_i) \\
& \text{subject to}
& & D(\hat{y_i})\leq \tau \\
\end{aligned}
\end{equation}
where $\tau\geq 0$, is referred to as the physical layer secrecy constraint which reflects the security performance of the system.
In our previous work~\cite{bhavyaphy}, we have derived the expressions for $D_{FC}$ and $D_{EV}$ (see Proposition $1$ and Proposition $2$ in~\cite{bhavyaphy}). Using those expressions \eqref{opt} reduces to:

\begin{equation}
\label{above}
\begin{aligned}
& \underset{\phi}{\text{maximize}}
& & \frac{\alpha(1-P_b \gamma)^2}{\gamma^2 P_t+\frac{\sigma^2}{\|\hat{P} s\|_2^2}}+(1-\alpha)\frac{\|\hat{P} s\|_2^2}{\sigma^2} \\
& \text{subject to}
& & \frac{(1-\alpha P_b \gamma)^2}{\gamma^2 P_t^E+\frac{\sigma^2}{\|\hat{P} s\|_2^2}}\leq \tau \\
\end{aligned}
\end{equation}
where
$\hat{P}=\phi^T(\phi \phi^T)^{-1}\phi$, 
$P_b=(P_1^0-P_2^0)+(P_2^1-P_1^1)$\\ 
$P_t=P_1^0+P_2^0-(P_1^0-P_2^0)^2$ and \\
$P_t^E=\alpha(P_1^0+P_2^0-\alpha(P_1^0-P_2^0)^2)$  matrix.
Next, we solve~\eqref{above} under various assumptions on the signal structure (e.g., known, low dimensional or sparse).

\section{Optimal Measurement Matrix Design with Physical Layer Secrecy Guarantees}

First, we explore some properties of the deflection coefficient at the FC, $D_{FC}$, and at the eavesdropper, $D_{EV}$, which will be used to simplify the measurement matrix design problem.

\begin{proposition}
Deflection coefficient both at the FC and the Eve is a monotonically increasing function of $D_H=\frac{\|\hat{P} s\|_2^2}{\sigma^2}$.
\end{proposition}
\begin{proof}
The proof follows from the fact that both $\frac{dD_{FC}}{dD_H}>0$ and $\frac{dD_{EV}}{dD_H}>0$.
\end{proof}

The above observation leads to the following equivalent optimal measurement matrix design problem for compressive detection:

\begin{equation}
\label{optimization}
\begin{aligned}
& \underset{\phi}{\text{maximize}}
& & \delta=\|\hat{P} s\|_2^2 \\
& \text{subject to}
& & \|\hat{P} s\|_2^2\leq \frac{\sigma^2}{\frac{(1-\alpha P_b\gamma)^2}{\tau}-\gamma^2P_t^E} \\
\end{aligned}
\end{equation}
for any arbitrary signal $s$. Note that, for the random measurement matrix $\delta_r=\|\hat{P} s\|_2^2=\frac{M}{N}\|s\|_2^2$~\cite{dave}. The factor $M/N$ can be seen as the performance loss due to compression as random measurement matrix fails to exploit the signal structure that may be known
\textit{a priori}. To improve performance, we consider the optimization of the measurement matrix by exploiting the signal structure while guaranteeing a certain level of secrecy. We show that any arbitrary secrecy constraint can be guaranteed by properly choosing the measurement matrix. 

\subsection{Known Signal Detection}
First, we consider the case where $s$ is known.

\begin{lemma}
When $s$ is known, the optimal value of the objective function of \eqref{optimization}, is given by $\delta^*=\min\left(\| s\|_2^2,\frac{\sigma^2}{\frac{(1-\alpha P_b \gamma)^2}{\tau}-\gamma^2 P_t^E}\right)$. 
\end{lemma}
\begin{proof}
The proof follows from the fact that $\hat{P}$ is an orthogonal projection operator, thus, $\|\hat{P}s\|_2^2\leq\|s\|_2^2$.
\end{proof}
Let us denote the singular value decomposition of $\phi=U[\pi_M,0]V^T$ where $U$ is an $M\times M$ orthonormal matrix, $[\pi_M,0]$ is an $M\times N$ diagonal matrix and $V$ is an $N\times N$ orthonormal matrix. Now, the optimal $\phi$ which achieves $\delta^*$ is characterized in the following lemma. 

\begin{lemma}
When $s$ is known, the optimal $\phi$ which achieves $\delta^*$ in~\eqref{optimization} is given by
$\phi^*=U[\pi_M,0](V^*R)^{T}$ where $U$ and diagonal $\pi_M>0$ are totally
arbitrary,
\begin{small}
$$R=\left[ \begin{array}{ccc}
\cos \theta& \mathbf{0} & \sin \theta\\
\mathbf{0} & \mathbf{I} &\mathbf{0}\\
-\sin \theta &\mathbf{0} &\cos \theta \end{array} \right],$$ \end{small} $\theta$ is the parameter which controls the level of secrecy such that $\theta=0$ if $\| s\|_2^2\leq\frac{\sigma^2}{\frac{(1-\alpha P_b \gamma)^2}{\tau}-\gamma^2 P_t^E}$, and, $\theta=\cos^{-1}\sqrt{\frac{\sigma^2/\| s\|_2^2}{\frac{(1-\alpha P_b \gamma)^2}{\tau}-\gamma^2 P_t^E}}$, otherwise. $V^*=[v_1^*,\cdots,v_N^*]$ is any orthonormal matrix satisfying $v_i^* \perp s, \forall i > M$.
\end{lemma}
\begin{proof}
To prove the lemma, notice that
\begin{small}$$\|\hat{P} s\|_2=s^TV\left[ \begin{array}{ccc}
I_M & 0 \\
0 & 0 \end{array} \right]V^Ts=\sum_{i=1}^{M} \tilde{s}_i^2\leq \| s\|_2^2$$\end{small}
where $\tilde{s}=V^T s$. The upper bound or equality in the above equation can be achieved if and only if $\tilde{s}_i=0,\;\forall i>M$. The corresponding optimal measurement matrix for
this case is characterized by
 $\phi^{*}=U[\pi_M,0]V^T$
where the orthonormal $U$ and diagonal $\pi_M>0$ are totally
arbitrary, while $V=[v_1,\cdots,v_N]$, as seen above, has to be an orthonormal matrix
satisfying $v_i \perp s, \forall i > M$.
Now, the matrix $(VR)$ is a orthonormal matrix for any orthonormal $V$ and observe that the optimal $V^*$ as given in the lemma is also orthonormal. Thus, for optimal $\phi^*$, we have
$$\|\hat{P} s\|_2=s^TV^*R\left[ \begin{array}{ccc}
I_M & 0 \\
0 & 0 \end{array} \right](V^*R)^Ts=\cos^2\theta\|s\|_2^2.$$
Next, using the definition of $\theta$, the results in the lemma can be derived. 
\end{proof}

If we define $\mathrm{Proj}_{u}(w)=\frac{u^Tw}{u^Tu}u$ and $W=[w_1,\cdots,w_N]$ with $w_1=s$ and $w_k$ as any linearly independent set of vectors, one possible solution for $V^*$ in a closed form is:
$V^*=[v_1,\cdots,v_N]$, where $v_k=\frac{u_k}{\|u_k\|_2}$ and $u_k=w_k-\sum_{j=1}^{k-1}\mathrm{Proj}_{u_j}(w_k).$ 
Note that, without physical layer secrecy constraint (or when $\theta=0$) the optimal value of the objective function is $\|s\|_2^2$. Thus, there is no performance loss due to compression. With physical layer secrecy constraint, $\theta$ serves as a tuning parameter to guarantee a certain level of secrecy.
This approach provides the optimal measurement matrix with a secrecy guarantee for a known $s$. However, in certain practical scenarios we do not have an exact knowledge of $s$. Next, we consider the cases where $s$ is not completely known. 

\subsection{Low Dimensional Signal Detection}
In this subsection, we consider the case where $s$ is not completely known but is known to lie in a low dimensional subspace and design $\phi$ so that  the detection performance at the FC is maximized while ensuring a certain level of secrecy at the eavesdropper.
We assume that $s$ resides in a $K$-dimensional subspace where $K<N$. That is to say, $s$ can be expressed as $s=D\beta$
where $D$ is an $N\times K$ matrix, whose columns are orthonormal, and $\beta$ is the $K\times1$ signal vector. Without loss of generality, we assume that $\|\beta\|_2^2=1$. Next, we look at the following two cases: $1)$ $D$ can be designed, $2)$ $D$ is fixed and known. For both the cases, we assume that $\beta$ is deterministic but unknown and find $\phi$ which maximizes the worst case detection performance. Formally, for the case where $D$ is a design parameter, the problem can be stated as

\begin{equation}
\label{opt2}
\begin{aligned}
& \underset{\phi_{M\times N}}{\max}\;\;\underset{D_{N\times K}}{\max}\;\;\underset{\beta_{K\times 1}}{\min} 
& & \delta=\|\hat{P} D\beta\|_2^2 \\
& \text{subject to}
& & \|\beta\|_2^2= 1,\;\|\hat{P} D\beta\|_2^2\leq \Delta \\
\end{aligned}
\end{equation}
where $\Delta=\frac{\sigma^2}{\frac{(1-\alpha P_b\gamma)^2}{\tau}-\gamma^2P_t^E}$.

We state the Courant-Fischer theorem which will be used to solve the above optimization problem. 

\begin{theorem}{(Courant-Fischer\cite{horn})}
Let $A$ be a symmetric matrix with eigenvalues $\lambda_1\geq\cdots\geq\lambda_N$ and $S$ denote the any $j$-dimensional linear subspace of $\mathbb{C}^N$. Then,
$$\underset{S:\;\text{dim}(S)=j}{\max}\;\;\underset{x\in S}{\min}\;\;\dfrac{x^TAx}{x^Tx}=\lambda_j.$$
\end{theorem}

\begin{lemma}
When $s$ lies in a low dimensional signal subspace, the 
optimal value of the objective function of \eqref{opt2} is given by $\delta^*=\min\left(\| \beta\|_2^2,\Delta\right)$ if $K\leq M$, and $\delta^*=0$, otherwise. 
\end{lemma}
\begin{proof}
Using Courant-Fischer theorem, we can show that the problem~\eqref{opt2} without a physical layer secrecy constraint is equivalent to $\underset{\phi_{M\times N}}{\max} \lambda_k(\hat{P})$.  
Now, the proof follows by observing that $\hat{P}$ is the orthogonal projection operator and its eigenvalues are given by $\lambda_i=1$ for $i=1$ to $M$ and $\lambda_i=0$ for $i=M+1$ to $N$.
\end{proof}
Next, we assume that $K\leq M$ and characterize the optimal measurement matrix $\phi^*$ and the optimal subspace $D^*$.

\begin{lemma}
\label{lem4}
When $s=D\beta$ and $K\leq M$, the optimal $(\phi^*,D^*)$ which achieves $\delta^*$ should satisfy the following condition: for any arbitrary $\phi=U[\pi_M,0]V^T$ where $V=[v_1,\cdots,v_N]$, the optimal $D^*=\cos\theta D$ with $D=[v_1,\cdots,v_K]$.
\end{lemma}
\begin{proof}
Note that for optimal $(\phi^*,D^*)$, we have
\begin{footnotesize}
\begin{eqnarray*}
\|\hat{P} s\|_2&=& \beta^T (D^*)^TV^*\left[ \begin{array}{ccc}
I_M & 0 \\
0 & 0 \end{array} \right]((D^*)^TV^*)^T\beta\\
&=& \beta^T \left[ \begin{array}{ccc}
\cos\theta I_K & 0 \end{array} \right]
\left[ \begin{array}{ccc}
I_M & 0 \\
0 & 0 \end{array} \right]
\left[ \begin{array}{ccc}
\cos\theta I_K  \\
0  \end{array} \right]\beta\\
&=&(\cos\theta)^2\sum\limits_{i=1}^{\min(K,M)}(\beta_i)^2
\end{eqnarray*}
\end{footnotesize}

Observe $\underset{\beta}{\min}\sum\limits_{i=1}^{\min(K,M)}(\beta_i)^2=\| \beta\|_2^2$ if $\min(K,M)=K$ and $0$, otherwise.
Using the definition of $\theta$, $\delta^*$ can be achieved. 
\end{proof}
The above lemma can be interpreted as follows: for any fixed $\phi$, one can choose $D$ accordingly, so that the upper bound $\delta^*$ can be achieved. Next, we look at the case where $D$ is fixed and we only optimize measurement matrix $\phi$. Observe that,
$$ \underset{\phi}{\max}\;\;\underset{\beta}{\min} 
\|\hat{P} D\beta\|_2^2\leq \underset{\phi}{\max}\;\;\underset{D}{\max}\;\;\underset{\beta}{\min} 
\|\hat{P} D\beta\|_2^2=\| \beta\|_2^2.$$
For a fixed $D$, the optimal value $\delta^*$ of the problem~\eqref{opt2} serves as an upper bound. 
To simplify the problem, we introduce an $(N\times N)$ matrix $P$ to guarantee secrecy in the system. In other words, $y_i=\phi P u_i$ where $P$ is determined to guarantee physical layer secrecy. Next, we find $\phi$ for which this upper bound is achievable for a fixed $D$ and $P$ to secrecy.

\begin{lemma}
For the low dimensional signal case $y_i=\phi P s$ with $s=D\beta$ where $D=[d_1,\cdots,d_K]$ is orthonormal,
the optimal measurement matrix $(\phi^*,P^*)$, is given by $P^*=\cos\theta I_{N\times N}$ and $\phi^*=U[\pi_M,0](V^*)^T$ where the orthonormal $U$ and diagonal $\pi_M>0$ are totally arbitrary, while $V^*=[v_1,\cdots,v_N]$ is such that $v_i=d_i$ for $i=1$ to $K$ and $v_i$ for $i=K+1$ to $N$ are such that $V$ forms an orthonormal basis.
\end{lemma}
\begin{proof}
The proof is similar to Lemma~\ref{lem4}, thus, omitted.
\end{proof}
For both the cases, where $D$ can be designed and where $D$ is fixed and known, without secrecy constraint the optimal value of the objective function is $\|s\|_2^2$. Thus, there is no performance loss due to compression. With secrecy constraint, $\theta$ serves as a tuning parameter to guarantee a certain level of secrecy.

\subsection{Sparse Signal Detection}
\vspace{-0.05in}

In this section, we assume that $s$ is $K$-sparse in the standard canonical basis and $\|s\|_2^2=1$. Also,
the exact number of the nonzero entries in $s$, their locations, and their values are assumed to be unknown. We design $\phi$ which maximizes the worst case detection performance by employing a lexicographic optimization approach\footnote{We first find a set of solutions that are optimal for a $k_1$-sparse signal. Then, within this set, we find a subset of solutions that are also optimal for $(k_1+1)$-sparse signals. This approach is known as a lexicographic optimization.}. Formally, the problem is
\vspace{-0.06in}

\begin{footnotesize}
\begin{equation}
\label{opt3}
\begin{aligned}
& \underset{\phi_{M\times N}}{\max}\;\;\underset{s}{\min} 
& & \|\hat{P} s\|_2^2 \\
& \text{subject to}
& & \|s\|_2^2= 1,\;\|s\|_0= K, \\
& & &\|\hat{P} s\|_2^2\leq \Delta,\; \phi \in \mathcal{A}_{K-1}\\
\end{aligned}
\end{equation}
\end{footnotesize}

\noindent where $\mathcal{A}_{K-1}$ is the set of solutions to the above optimization problem for sparsity level $K-1$ and $\Delta$ is defined in~\eqref{opt2}.

\begin{lemma}
There is no performance loss while solving the problem~\eqref{opt3} if we restrict our solution space to be matrices on the Stiefel manifold $S_t(M,N)$, where $$S_t(M,N):=\{\phi\in\mathbb{R}^{M\times N}:\phi\phi^T=I\}.$$
\end{lemma}
\begin{proof}
The proof follows from the observation that $\pi_M=I_M$ for frames $\phi=U[\pi_M,0]V^T$ in Stiefel manifold and the value of $\|\hat{P}s\|_2^2$ is independent of $\pi_M$ and $U$.
\end{proof}

Next, we limit our focus on Stiefel manifolds and establish an upper bound on the value of the objective function in~\eqref{opt3} for different sparsity levels. Later we find measurement matrices which can achieve this upper bound.  
\begin{lemma}
\label{lem7}
For the sparsity level $K=1$, the optimal value of the objective function of~\eqref{opt3} is $\min\left(\frac{M}{N},\Delta\right)$. For the sparsity level $K\geq 2$, an upper bound on the value of the objective function is given by $\min\left(\frac{M}{N}(1-\mu),\Delta\right)$, where $\mu=\sqrt{\frac{N-M}{M(N-1)}}$.
\end{lemma}
\begin{proof}
The proof is similar to Theorem $1$ and Theorem $3$ as given in~\cite{Zahedi}, thus, omitted.
\end{proof}
\begin{lemma}
The optimal measurement matrix $(\phi^*,P^*)$, for the $K$ sparse signal case $y_i=\phi P s$ is given by:
\begin{itemize}
\item For the sparsity level $K=1$, $\phi^*$ is a uniform tight frame with norm values equal to $\sqrt{M/N}$ and $P^*=\cos\theta I_{N\times N}$,
\item For the sparsity level $K\geq2$, $\phi^*$ is an equiangular tight frame with norm values equal to $\sqrt{M/N}$ and $P^*=\cos\theta I_{N\times N}$.
\end{itemize}
\end{lemma}
\begin{proof}
Proof follows from the definition of uniform (or equiangular) tight frames~\cite{Casazza} and observation that the upper bounds in the Lemma~\ref{lem7} can be reached only by these frames.
\end{proof}
Note that, without physical layer secrecy constraint (i.e., $\theta=0$), our results reduce to the ones in~\cite{Zahedi}. With physical layer secrecy constraint, similar to previous cases, $\theta$ serves as a tuning parameter to guarantee an arbitrary level of secrecy.
Also, it is shown that a real equiangular
tight frame can exist only if $N \leq M(M+1)/2$, and a complex equiangular tight frame requires $N \leq M^2$~\cite{etf}. When $M$ and $N$ do not satisfy this condition, the bound in Lemma~\ref{lem7} can not be achieved and one can employ a heuristic or algorithmic approach~\cite{bai}.
\section{Discussion and Future Work}
We considered the problem of designing measurement matrices for high
dimensional signal detection based on compressed
measurements with physical layer secrecy guarantees. 
It was shown that the optimal design depends on the nature of the signal to be detected. Further, security performance of the system can be improved by using optimized measurement matrices
along with artificial noise injection based techniques. In the future, we plan to come up with efficient algorithms to improve the worst-case detection probability for the cases where equiangular tight frames do not exist.
\bibliographystyle{IEEEtran}
\bibliography{Conf,Book,Journal}
\end{document}